\documentclass[runningheads, orivec]{llncs}
\usepackage[T1]{fontenc}
\usepackage{graphicx}
\usepackage{comment}
\usepackage{xcolor}
\usepackage[labelformat=simple]{subcaption}

\usepackage{amsmath}
\usepackage{amsfonts}

\usepackage{acronyms}

\usepackage[numbers,sort&compress]{natbib}

\begin{document}
\title{MEV Capture Through Time-Advantaged Arbitrage}

\author{}
\institute{}

%
%
\author{Robin Fritsch\inst{1, 2} \and Maria Inês Silva\inst{2, 5} \and Akaki Mamageishvili\inst{3} \and Benjamin Livshits\inst{4} \and Edward W. Felten\inst{3}}
%
\authorrunning{Fritsch et al.}
%
\institute{ETH Zurich \and Matter Labs \and Offchain Labs \and Imperial College London, \and NOVA Information Management School}
%
\maketitle              
\begin{abstract}
As blockchains begin to process significant economic activity, the ability to include and order transactions inevitably becomes highly valuable, a concept known as Maximal Extractable Value~(MEV). This makes effective mechanisms for transaction inclusion and ordering, and thereby the extraction of MEV, a key aspect of blockchain design. Beyond traditional approaches such as ordering in a first-come-first-serve manner or using priority fees, a recent proposal suggests auctioning off a time advantage for transaction inclusion. In this paper, we investigate this time advantage mechanism, focusing specifically on arbitrage opportunities on Automated Market Makers~(AMMs), one of the largest sources of MEV today. We analyze the optimal strategy for a time-advantaged arbitrageur and compare the profits generated by various MEV extraction methods. Finally, we explore how AMMs can be adapted in the time advantage setting to capture a portion of the MEV.

\end{abstract}

\section{Introduction}

\gls{DeFi} has played a central role in the Ethereum ecosystem ever since its inception, encompassing many of its most-used applications. DeFi encompasses a range of financial services, including token trading, as well as borrowing and lending, which are facilitated through smart contracts deployed on the blockchain.
More recently, \gls{DeFi} has begun expanding to \gls{L2} scaling solutions. Trading volumes on \glspl{L2} are rising consistently, with some rollups already processing more token swaps than the Ethereum itself~\cite{2024DeFiCategories,2024L2BeatTVL}.

Two main features of rollups drive this trend in \gls{DeFi}. First, rollups offload most of the computation required to verify and execute transactions outside the Ethereum network, only storing condensed batches of transactions on Ethereum. This allows them to run more efficiently and reduce transaction costs for end users. Second, rollups store the condensed batches of transactions and proofs on Ethereum, thus inheriting (some of) Ethereum's security guarantees.

As \gls{DeFi} activity on rollups grows, so do the incentives for various opportunistic economic practices collectively known as \gls{MEV}. \gls{MEV} -- a phenomenon first described for the Ethereum blockchain~\cite{daian_flashboys_2020} -- describes the value that can be extracted from the ability to include and reorder transactions on a blockchain. Currently, the most popular forms of \gls{MEV} are arbitrages (that try to capitalize on price discrepancies between different on-chain trading venues), liquidations (where collateral from debt repayment in lending protocols can be purchased at a discount), and sandwiching (which involves wrapping a victim's swap transaction between two new transactions in such a way that provides a worse trade execution to the victim while making a small profit)~\cite{qin_quantifying_mev_2022}.
On the most prominent rollups today, arbitrages and liquidity occur in similar volumes to Ethereum, while sandwich attacks have not been observed \cite{torres_rolling_2024}.

\subsection{MEV Capture Mechanisms}

With both \gls{DeFi} and arbitrage activity increasing, policies to include and order transactions have become a key design choice for rollups. Initially, most rollups employed a simple design to order transactions solely based on when the centralized sequencer received them. This ordering system is known as \gls{FCFS}, and both optimistic (such as Arbitrum) and ZK (such as ZKsync) rollups still use it. In this system, transaction fees are determined by the state transition function, and users cannot pay a premium for faster inclusion. Instead, users who value fast transaction inclusion are motivated to invest in latency infrastructure in order to improve their latency to the sequencer.

Another popular mechanism is a \glspl{PGA}.
In this scheme, transactions are grouped into blocks which are produced regularly.
Users pay a priority fee to the block builder, according to which transactions in a block are ordered.
With this design, users benefit less from a latency advantage and instead express the value of their transaction by paying higher priority fees. Rollups such as Optimism and Base employ this design. This is also the design used by Ethereum, even though most MEV extraction is currently performed through the MEV-Boost implementation.

As expected, the transaction ordering policy significantly influences how the profit from \gls{MEV} extraction is distributed among participants. For instance, on Ethereum, most profits are paid to validators through MEV-Boost, while on Arbitrum and for \gls{FCFS}, more generally, profits stay with \gls{MEV} extractors and are likely invested in improving their access to sequencers (e.g., through latency improvements). Avoiding such latency competitions would arguably lead to a more efficient market design. This can be achieved by letting extractors bid for the opportunity to extract \gls{MEV} in an auction.
A recent proposal by Arbitrum~\cite{arb_timeboost_2024} involves introducing an auction market to sell a time advantage, with the profits going to the rollup protocol. Concretely, users bid for the right to access an ``express lane'' where transactions are sequenced immediately. This access is sold to a single user for a fixed time interval. Transactions submitted by the remaining users are artificially delayed by the sequencer for a predefined time, resulting in the user who purchases the time advantage being able to guarantee their transactions being sequenced ahead of its competitors.

The key questions are how rational \gls{MEV} extractors will behave under this design and how much value rollups are expected to capture.
In particular, the new mechanism adds a new layer of complexity: the timing of arbitrage. Previously, competition forced arbitrageurs to act as soon as opportunities were detected. However, a time-advantaged arbitrageur possesses an exclusive time window to react and now faces a choice: take an early profit, fearing the opportunity may close, or wait in hopes of securing a larger future profit.

Moreover, the question arises whether such a design introduces new opportunities for \gls{MEV} extraction and, thus, higher profits for extractors at the expense of other participants. In this work, we aim to tackle these questions.

\subsection{Contributions}

This work studies and compares emergent sequencing policies for rollups. To the best of our knowledge, this is the first that not only measures the impact of sequencing rules on arbitrage profits but also analyses how introducing a time advantage to the sequencing policy impacts arbitrageurs' decisions. Concretely, this work provides five main contributions:
\begin{itemize}
    \item A theoretical model to analyze the scenario where a single actor has a time advantage and uses it to perform arbitrage between an \gls{AMM} liquidity pool and an (infinitely liquid) external venue. This model assumes a sequencing design similar to the current Timeboost proposal from Arbitrum~\cite{arb_timeboost_2024}.
    \item A theoretical derivation of the optimal strategy for when the advantaged arbitrageur submits the extraction transaction under the proposed model. 
    \item Using dynamic programming and the empirical price distributions of some key liquidity pools, we analyze the optimal strategy for when the advantaged arbitrageur submits the extracting trade. We find that waiting until the end of the advantage interval is the best strategy, assuming the future price distribution is independent of the past price movement.
    \item We simulate the expected profits an arbitrageur could extract in some key liquidity pools under three separate sequencing regimes -- \gls{FCFS}, \gls{PGA}, and Timeboost. We find that the theoretical finding in the previous step, that it is always optimal to wait with the trade as long as possible, is sometimes invalidated. A potential explanation of this finding is that future price distribution is not independent of the past price movement, which can actually be seen as an outside market imperfection.
    \item We look into an option of letting pools capture some share of the value from the arbitrage opportunity by letting the contract know about Timeboost transaction at the time of execution. This results in a sequential game where the pool first sets up a fee structure for extracting the value, and then the arbitrageur best responds to it by choosing the price that maximizes its returns. In the equilibrium of this game, the pool obtains 25\% and the time-advantaged arbitrageur obtains 50\% of the total value.
\end{itemize}

\subsection{Related Work}

Arbitrage in the Ethereum ecosystem has received a growing interest from researchers. On one hand, it is one of the most prevalent \gls{MEV} strategies. On the other hand, it has a significant impact on how \glspl{DEX} operate and, consequently, on the experience of both users and liquidity providers.

A key area of research is quantifying the profit potential of arbitrage and the amount actually extracted by arbitrageurs. \citet{heimbach2024nonatomic} study non-atomic arbitrage, a form of \gls{MEV} that exploits price differences between \glspl{DEX} and off-chain exchanges. The authors conservatively estimate that non-atomic arbitrage has historically accounted for about 30\% of trading volume on Ethereum's largest \glspl{DEX}. Our work also focuses on this type of non-atomic arbitrage but with two key distinctions. First, we assess the full potential arbitrage that could optimally be extractable -- as opposed to what has been extracted historically. Second, we do so in a setting with a latency auction, which introduces the problem of deciding the optimal timing for arbitrating. In contrast to \gls{FCFS}, where competition leads to arbitrageurs exploiting arbitrage opportunities as soon as they occur, or \glspl{PGA}, where arbitrageurs act as soon as the next block is produced, a time-advantaged arbitrageur can delay exploitation without the risk of another arbitrageur seizing the opportunity. In general, MEV and AMM arbitrage have to date solely been studied in either FCFS or PGA settings in the existing literature, including the works cited below.

\citet{milionis2022automated} study profits from arbitrage between \glspl{AMM} and an external exchange from a different angle: they consider them as losses to \glspl{LP} and coin the term \gls{LVR}. In follow-up work, an analytical expression for the arbitrage profits as a function of trading fees and block times is derived \cite{milionis2023automated}.
Following this work, \citet{fritsch2024arblosses} study historical arbitrage profits using the same \gls{LVR} formalization. In addition to measuring historical arbitrage losses across different Uniswap pools on Ethereum, the authors also analyze the relationship between block times and arbitrage profits.

Within the general scope of measuring \gls{MEV} profits, multiple works \cite{qin_quantifying_mev_2022, torres_rolling_2024} have considered rollups specifically (as opposed to previous works that focus on Ethereum). Both papers consider arbitrages as the most relevant form of \gls{MEV} but also consider liquidations and front-running/sandwiching. These analyses look at historical transaction data from the top rollups and define heuristics to detect common \gls{MEV} attacks. However, since this type of arbitrage is challenging to detect, neither work reports on non-atomic arbitrage between \glspl{DEX} and off-chain exchanges. 
Finally, \citet{gogol_layer2_arbitrage_2024} also looks into non-atomic arbitrages in rollups. However, they focus on quantifying the value of historically unexploited arbitrage opportunities between top Uniswap pools and Binance. In addition, they measure the average times for these opportunities to disappear.

\section{Model Description}\label{sec:model}

We study the scenario in which one single actor has a time advantage of $T_w$ over everybody else when submitting transactions to a rollup (or blockchain) sequencer over a time period of length $T$.
This time advantage is achieved by delaying the processing of all other actors' transactions by $T_W$, while the advantaged actor’s transactions are sequenced immediately. In other words, one actor purchases access to a ``fast lane'', allowing their transactions to bypass an artificially applied delay of $T_w$ during the regular processing of transactions.
Apart from the applied delay, transactions are ordered on a first-come, first-served basis. Note that there is no limit on the number of transactions that can be submitted through the fast lane during the period $T$ of access.

To assess the value of this advantage, we examine the potential arbitrage profits that can be extracted with this time edge.
Specifically, we examine the profit potential from performing arbitrage trades between an \gls{AMM} pool on the rollup and an external venue on which the same token pair is traded. The external venue could be thought of as comprising the entire outside market.
We assume the price impact of trading on the external venue is negligible compared to the price impact on the \gls{AMM}.~\footnote{This is the case when the external venue is assumed to have infinitely deep liquidity, or at least significantly deeper liquidity than the \gls{AMM} pool.}

\subsection{Arbitrage Opportunities on AMMs}

We consider an \gls{AMM} pool with a trading fee $f$, facilitating trades between two assets. These pools collect token deposits from \glspl{LP} as reserves, against which traders can swap for a fee. Exchange rates are algorithmically determined -- e.g., by maintaining a constant product of the reserves, in the simplest case~\cite{angeris21uniswap}.

We consider a pool containing a risky asset and a numéraire, and assume the risky asset's price fluctuates on the outside market over time. The AMM price adjusts solely through arbitrage trades with the outside market, reflecting our assumption that price discovery occurs externally. We ignore uninformed ("noise") trades on the AMM.

Let $p$ represent the price of the risky asset on the external market relative to its price on the \gls{AMM}. Specifically, for an AMM price $P_0$ and an external market price $P$, we define $p = P/P_0$.
For example, $p = 0.9$ indicates that the external price is 10\% lower than the \gls{AMM} price. Initially, we assume the prices on the \gls{AMM} and the external market are equal, i.e., $p = 1.0$.

Arbitrageurs monitor price discrepancies between the external market and the \gls{AMM}. When an arbitrage opportunity arises, they compete to exploit it. Realizing arbitrage profits involves executing one trade each on the external market and the \gls{AMM}.
Notably, only the first arbitrageur to interact with the \gls{AMM} can perform a profitable arbitrage, as their trade moves the AMM pool price towards the external market price.

In the setting we analyze, the time-advantaged arbitrageur's trades are immediately processed by the \gls{AMM}, while transactions from other arbitrageurs are delayed. Consequently, a time window of length $T_w$ exists, starting when an arbitrage opportunity emerges, during which no other arbitrageur can respond to the opportunity due to the imposed delay.
This allows the time-advantaged arbitrageur to decide whether and when to capitalize on the opportunity during this time window, contrasting with a pure first-come, first-serve setting where competition forces immediate extraction.

Additionally, we assume that any arbitrage opportunities not exploited by the time-advantaged arbitrageur within a time of $T_w$ of their emergence will be seized by other arbitrageurs.
Note that this is a conservative assumption with respect to the profits of the time-advantaged arbitrageur.\footnote{It is not guaranteed that delayed arbitrageurs will always attempt to exploit arbitrage opportunities due to the risk of partly execution and changes of the external price during the delay period.}
Importantly, any strategy employed by the time-advantaged arbitrageur under these assumptions will remain effective regardless of the actions of the delayed arbitrageurs.

\subsection{The Time-Advantaged Arbitrageur}

Our objective is to determine the optimal strategy for the time-advantaged arbitrageur and the resulting profits.
A strategy defines when the arbitrageur chooses to execute an arbitrage, based on three factors: the current relative price difference $p$, the elapsed time $t_w$ within their advantage window of length $T_w$, and the elapsed time $t$ within the overall period $T$.
Formally, a strategy maps every possible tuple $(p, t_w, t)$ to an action in $\{\text{arbitrage}, \text{wait}\}$.

We model the overall time horizon $T$ as a series of discrete time steps of size $\Delta$.
At each time step -- i.e., at time $t=0,\Delta,2\Delta,\ldots$ -- the arbitrageurs, particularly the one with a time advantage, have the opportunity to submit trades to both the \gls{AMM} and the external market.

Given the relative price difference $p$ between the \gls{AMM} and the external market, we consider the \emph{maximal arbitrage trade}, i.e., the arbitrage trade that yields the maximal profit, and denote the profits from executing it as $profit(p)$. This profit is measured in the numéraire and is considered relative to the \gls{AMM} pool's value, which is also measured in the numéraire at the \gls{AMM} price before the AMM trade.
Note that if the price difference is too small, no profitable arbitrage trade will be possible due to the trading fees on the \gls{AMM} and, in this case, $profit(p)=0$.
The profit function for a constant product AMM pool is derived in Section \ref{ssec:cpmm} and plotted in Figures \ref{fig:profit_function}.

Due to the \gls{AMM}'s trading fee, the maximal arbitrage does not completely eliminate the price difference. Let $pArb(p)$ represent the remaining relative price difference after the maximal arbitrage trade is executed at the relative price difference $p$.
Finally, let $poolVal(p)$ denote the relative change in the pool's value resulting from the maximal arbitrage trade executed at the relative price difference $p$.

Let $opt(p, t_w, t)$ represent the expected profit a time-advantaged arbitrageur can achieve between time $t$ and $T$ by employing the optimal strategy, given a relative price difference of $p$ at time $t$ and an elapsed time $t_w$ within the advantage window. That is,
 \begin{align*}
    opt\colon \mathbb{R}_{>0}\times [0,T_w]\times [0,T] &\to \mathbb{R}_{\geq 0}.
\end{align*}
Note that this definition of $opt()$ assumes that the external price of the risky asset follows a memoryless process, which means that future price changes are independent of past prices. Consequently, $opt()$ depends solely on the current relative price difference $p$ rather than the entire history of past prices.

At the end of the time period (i.e., $t=T$), seizing any remaining arbitrage opportunity is clearly optimal, as the arbitrageur’s time advantage ceases beyond this point. Therefore, we have
\begin{align*}
    opt(p,t_w, T) = profit(p)\qquad \forall\, p\in \mathbb{R}, t_w\in [0,T_w].
\end{align*}

Furthermore, we assume that it is always optimal for the arbitrageur to exploit any available arbitrage opportunity at the end of the delay time $T_w$ (since the opportunity will otherwise be taken by other arbitrageurs). When doing so, the arbitrageur instantly makes a profit of $profit(p)$ and resets their advantage window $t_w=0$ (since the other arbitrageurs need to react to the change of the AMM price).
Additionally, the arbitrageur expects further future profits between time $t$ and $T$.
Since, the arbitrage trade has moved the relative price difference to $pArb(p)$, the expected future profits are given by $opt(pArb(p),0,t)$.
However, the future profits are measured relative to the pool value after the arbitrage. So, this term needs to be multiplied by the relative change in the pool value $poolVal(p)$.
Combining immediate profits and scaled future profits, the total expected profit is:
\begin{align*}
    opt(p,T_w, t) = profit(p) + poolVal(p)\cdot opt(pArb(p),0,t) \qquad \forall\, p\in \mathbb{R}, t\in [0,T].
\end{align*}

For any other point in time, that is, when $t_w\neq T_w$ and $t\neq T$, the arbitrageur must decide between two actions:  
\begin{enumerate}
    \item Execute an arbitrage trade immediately at time $t$.
    \item Wait until the next time step $t + \Delta$. 
\end{enumerate}
The arbitrageur will chose the action which yields the higher expected future profits.

If the arbitrageur chooses to wait, the expected future profit depends on the expected relative price difference at the next time step.
Let $p_\Delta$ be the random variable describing the relative price change over the time interval $\Delta$.
The expected profit of waiting is then calculated as the profit from the optimal strategy at the next step, given the updated price $p \cdot p_\Delta$ and the updated times $t_w + \Delta$ and $t + \Delta$:
\begin{align*}
    wait(p, t_w, t) := \mathbb{E}\left[opt(p\cdot p_\Delta, t_w+\Delta, t+\Delta) \right].
\end{align*}
Here, the expectation is taken over the random variable $p_\Delta$.

On the other hand, the expected profit obtained by arbitraging is given by the same formula as in the previously considered case of arbitraging at the end of the time window. Concretely, it is given by
\begin{align*}
    arb(p, t_w, t) := profit(p) + poolVal(p)\cdot opt(pArb(p),0,t).
\end{align*}
Therefore, the advantaged arbitrageur's optimal expected profit is given by the maximum of the expected profit of their two possible decisions:
\begin{align*}
    opt(p,t_w,t) &= \max\big(wait(p,t_w,t), arb(p,t_w,t)\big)\\    
    &= \max\Big( \mathbb{E}\left[opt(p\cdot p_\Delta, t_w+\Delta, t+\Delta) \right],\\
    &\qquad\qquad\qquad\qquad profit(p) + poolVal(p)\cdot opt(pArb(p),0,t) \Big).
\end{align*}
Finally, to obtain a recursive definition, we express $opt(pArb(p),0,t)$ in terms of results from times steps later than $t$. This is needed for the empirical analysis later on. Since, by definition of $pArb(p)$, no profitable arbitrage trade is possible in this situation, the arbitrageur's only option is to wait for a time step $\Delta$. Therefore,
\begin{align*}
    opt(pArb(p),0,t) = wait(pArb(p),0,t) =\mathbb{E}\left[opt(pArb(p)\cdot p_\Delta, 0+\Delta, t+\Delta) \right].
\end{align*}
The four expressions above fully define $opt$ for all input values, given the $profit()$ function and the distribution of price changes $p_\Delta$.

\section{Optimal Strategy}

\subsection{Arbitrage Profits for Constant Product Pools}
\label{ssec:cpmm}

If the pool being arbitraged is a \gls{CPMM}, the profit function ($profit(p)$), and the remaining relative price difference and pool size after the arbitrage trade ($pArb(p)$ and $poolVal(p)$, respectively) in terms of the relative price difference $p$ and the trading fee $f$ is:
\hspace{-1em}
\begin{align}\label{eq:profit_arb}
    profit(p) = \begin{cases}
    \frac{1}{2}p - \frac{1}{\sqrt{1-f}}\sqrt{p} +\frac{1}{2}\frac{1}{1-f}  & \text{if } p > \frac{1}{1-f} \\
    0 & \text{if } 1-f \leq p \leq \frac{1}{1-f} \\
    \frac{1}{2}\frac{1}{1-f}p - \frac{1}{\sqrt{1-f}}\sqrt{p} +\frac{1}{2} & \text{if } p < 1-f
\end{cases}
\end{align}
\hspace{-2em}
\begin{align}\label{eq:price_arb}
    pArb(p) = \begin{cases}
    \frac{1}{1-f}  & \text{if } p > \frac{1}{1-f} \\
    p & \text{if } 1-f \leq p \leq \frac{1}{1-f} \\
    1-f & \text{if } p < 1-f
\end{cases}
\end{align}
\hspace{-2em}
\begin{align}
    poolVal(p) = \sqrt{\frac{p}{pArb(p)}} = \begin{cases}
    \sqrt{1-f} \sqrt{p}  & \text{if } p > \frac{1}{1-f} \\
    1 & \text{if } 1-f \leq p \leq \frac{1}{1-f} \\
    \frac{1}{\sqrt{1-f}}\sqrt{p} & \text{if } p < 1-f.
\end{cases}
\end{align}

For conciseness, we provide the full derivation of these formulas in Section \ref{annex:cpmm-profit} of the appendix.
Figure \ref{fig:profit_function} plots the profit function for a trading fee of $f=0.05\%$.

\begin{figure}[ht]
    \centering
    \includegraphics[width=0.7\linewidth]{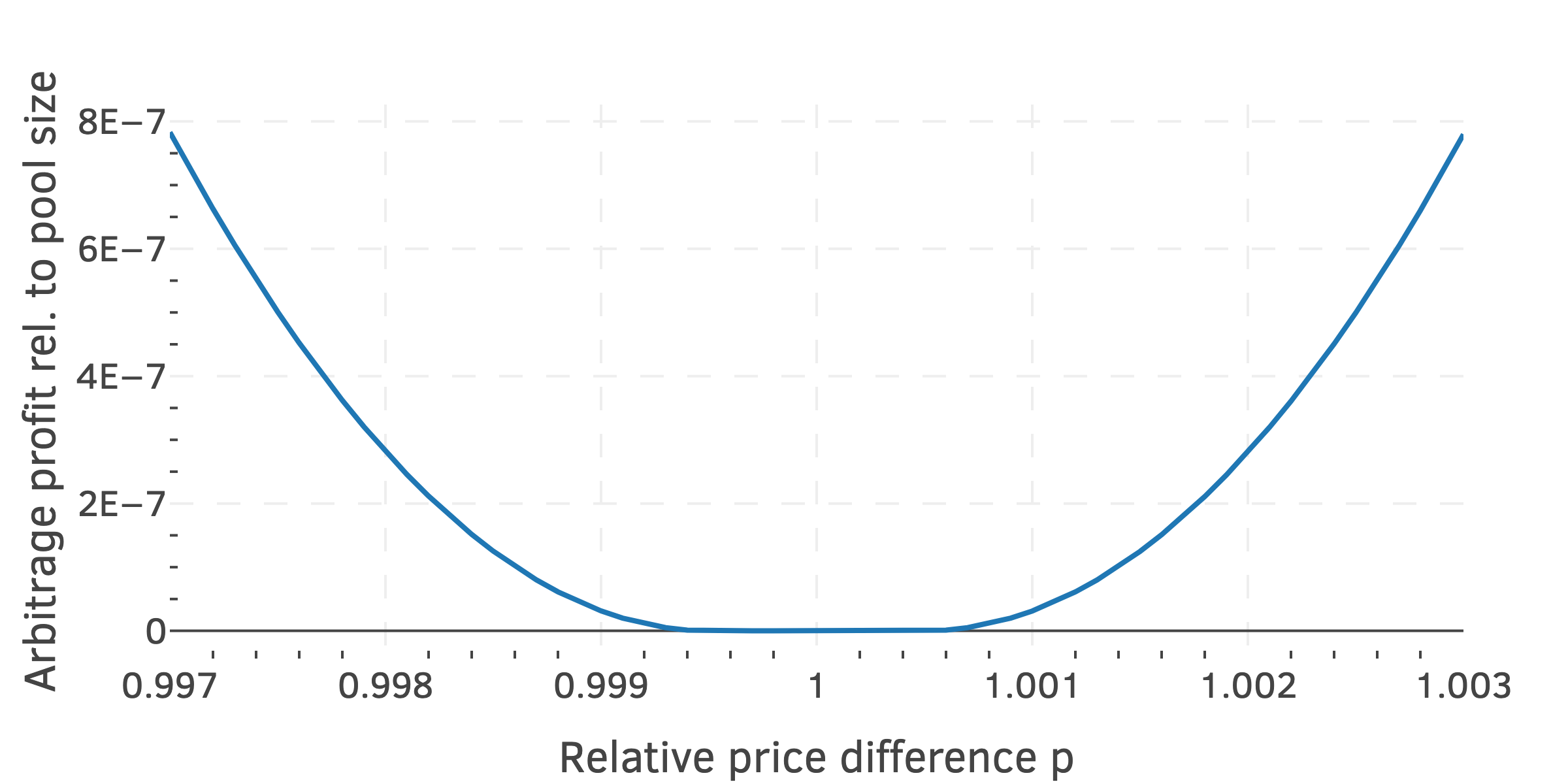}
    \caption{Maximal arbitrage profit relative to the pool value for a constant product pool with a trading fee of $f=0.05\%$.}
    \label{fig:profit_function}
\end{figure}

\subsection{Analytical Results}

The decision between arbitraging and waiting hinges on which of the two terms in the definition of $opt(p,t_w,t)$ is greater. Specifically, if $arb(p,t_w,t) > wait(p,t_w,t)$, then arbitraging is the more profitable option. Conversely, if $wait(p,t_w,t)$ exceeds $arb(p,t_w,t)$, waiting is the more advantageous action. Using the fact that $profit(pArb(p))=0$, the condition for when arbitraging is advantageous can be rewritten in the following way.

\begin{corollary}
    For the time-advantaged arbitrageur, arbitraging is advantageous over waiting if
    \begin{multline} \label{eq:arb_wait_condition}
        \mathbb{E}\left[opt(pArb(p)\cdot p_\Delta, 0+\Delta, t+\Delta) \right] - profit(pArb(p)) > \\
        \frac{1}{poolVal(p)}\left(\mathbb{E}\left[opt(p\cdot p_\Delta, t_w+\Delta, t+\Delta)\right] - profit(p)\right).
    \end{multline}
\end{corollary}

From this corollary, we can derive the optimal strategy for the time-advantaged arbitrageur, given certain assumptions. Concretely, we start by assuming that the pool has no trading fees ($f=0$) and that the expected future price equals the current price, i.e., $\mathbb{E}[p_\Delta]=1$.

In particular, the latter holds when prices follow a geometric Brownian motion -- the standard model for asset prices. There, the price changes $p_\Delta$ are distributed independently and log-normally, i.e.,
\begin{align*}
    p_\Delta \sim \exp\left(-\frac{\sigma^2}{2}\Delta + \sigma\mathcal{N}(0,\Delta)\right),
\end{align*}
where $\mathcal{N}(0,\Delta)$ denotes a Gaussian random variable.
To check this, use the expression for the expectation of a log-normal:
\begin{align*}
    \mathbb{E}[p_\Delta] = e^{-\frac{\sigma^2}{2}\Delta}e^{\frac{\sigma^2}{2}\Delta} = 1.
\end{align*}

\begin{proposition}\label{indifference}
    Consider a CPMM pool without trading fees ($f=0$) and a distribution of price changes such that the expected future price equals the current price, i.e., $\mathbb{E}[p_\Delta]=1$.
    Then, the time-advantaged arbitrageur is always indifferent to waiting or arbitraging.
\end{proposition}

\begin{proof}
We prove the statement by inductively computing $wait(p,t_w,t)$ and $arb(p,t_w,t)$ starting at the end of the time interval $T$.

For $f=0$, the function for the maximal arbitrage profit in \eqref{eq:profit_arb} simplifies, and the optimal profit at time $t=T$ is given by
\begin{align*}
    opt(p, t_w, T) = profit(p) = \frac{1}{2}p - \sqrt{p} + \frac{1}{2}.
\end{align*}
With this, condition \eqref{eq:arb_wait_condition} for time $t=T-\Delta$ simplifies to
\begin{multline} \label{eq:arb_wait_condition_profit}
    \mathbb{E}\left[profit(pArb(p)\cdot p_\Delta)\right] - profit(pArb(p)) > \\
    \frac{1}{poolVal(p)}\left(\mathbb{E}\left[profit(p\cdot p_\Delta)\right] - profit(p)\right).
\end{multline}
With this, it follows that
\begin{align*}
    \mathbb{E}\left[profit(p\cdot p_\Delta)\right] &= \mathbb{E}\left[\frac{1}{2}p\cdot p_\Delta - \sqrt{p\cdot p_\Delta} + \frac{1}{2}\right] \\
    &= \frac{1}{2}p\cdot \mathbb{E}\left[p_\Delta\right] - \sqrt{p}\cdot \mathbb{E}\left[\sqrt{p_\Delta}\right] + \frac{1}{2} \\
    &= \frac{1}{2}p - \sqrt{p}\cdot \mathbb{E}[\sqrt{p_\Delta}] + \frac{1}{2}.
\end{align*}
Inserting the above and the expression for $profit(p)$ into condition \eqref{eq:arb_wait_condition_profit} yields
\begin{align*}
    \sqrt{pArb(p)}\left(1-\mathbb{E}[\sqrt{p_\Delta}]
    \right) 
    > \frac{1}{\sqrt{\frac{p}{pArb(p)}}} \sqrt{p} \left(1-\mathbb{E}[\sqrt{p_\Delta}]\right).
\end{align*}
Since the two sides are equal, at time $t=T-\Delta$, the arbitrageur receives the same expected profit by arbitraging or waiting. In other words, one decision is not advantageous over the other.

Additionally, it follows that the optimal profit for $t=T-\Delta$ can be computed using the term for waiting:
\begin{align*}
    opt(p, t_w, T-\Delta) &= \mathbb{E}\left[opt(p\cdot p_\Delta, t_w+\Delta, T) \right]\\
    &= \mathbb{E}\left[profit(p\cdot p_\Delta) \right] = \frac{1}{2}p - \sqrt{p}\cdot \mathbb{E}[\sqrt{p_\Delta}] + \frac{1}{2}
\end{align*}

The previous arguments can now inductively be applied for~$t=T-2\Delta, T-3\Delta,\ldots$ to show that arbitraging and waiting also have the same expected profit at these times, and that
\begin{align*}
    opt(p, t_w, T-k\Delta) &= \mathbb{E}\left[opt(p\cdot p_\Delta, t_w+\Delta, T-(k-1)\Delta)\right] \\
    &= \mathbb{E}\left[\frac{1}{2}p\cdot p_\Delta - \mathbb{E}[\sqrt{p_\Delta}]^{k-1} \sqrt{p \cdot p_\Delta} + \frac{1}{2}\right] \\
    &= \frac{1}{2}p - \sqrt{p}\cdot \mathbb{E}[\sqrt{p_\Delta}]^{k} + \frac{1}{2}.
\end{align*}
\qed
\end{proof}
For a geometric Brownian motion where $\sqrt{p_\Delta}$ is log-normally distributed (since $p_\Delta$ is), and
\begin{align*}
    \mathbb{E}[\sqrt{p_\Delta}] = e^{-\frac{\sigma^2}{4}\Delta}e^{\frac{(\sigma/2)^2}{2}\Delta} = e^{-\frac{\sigma^2}{8}\Delta},
\end{align*}
it follows that the optimal expected profit over the whole period is
\begin{align*}
    opt(p, t_w, 0) = \frac{1}{2}p - \sqrt{p}\cdot e^{-\frac{\sigma^2}{8} T} + \frac{1}{2}.
\end{align*}

\subsection{Simulation Results}

The optimal strategy can also be computed numerically for various price change distributions $p_\Delta$. To achieve this, we discretize the range of possible price changes between consecutive time steps. Recall that both time $t$ and time within the window $t_w$ are already discretized in our model. Specifically, we consider equally spaced discrete values within an interval $[1-p_{max}, 1+p_{max}]$ as potential price changes.

Using this approach, the optimal strategy consists of the optimal decision -- whether to wait or arbitrage -- for each state described by the triples $(p,t_w,t)$, with $1-p_{max}\leq p\leq 1+p_{max}$, $0\leq t_w\leq T_w$, and $0\leq t\leq T$.
These decisions are based on the expected profits from waiting and arbitraging in each state, i.e., the values of $wait()$ and $arb()$ for the respective parameters.
These values can be computed using dynamic programming. The values for $t=T$ are straightforward, as arbitraging is the optimal option for the time-advantaged arbitrageur at the final time step. The values for time $t=T-\Delta$ can then be computed recursively using the values for $t=T$ following the formulas provided in Section \ref{sec:model}, and this process continues backward through time.

We run the simulations using the following parameters, which are intended for implementation on Arbitrum:\footnote{Experimenting with various parameter values did not produce significantly different results.} $\Delta=10$ ms, $T_w=200$ ms, $T=1$ min, and consider commonly used values for the pool's trading fee, including 0.05\% and 0.3\%.
For the price of the volatile token, we analyze the following two scenarios.

\paragraph{Geometric Brownian motion.}

When prices follow a geometric Brownian motion -- the standard model for asset prices -- prices changes in subsequent time intervals are independent and log-normally distributed.

We simulate the optimal strategy for this type of price distribution, assuming zero drift and a daily volatility ranging from 1\% to 10\%. Our results show that the best strategy for the time-advantaged arbitrageur is to wait for the entire duration of their time advantage. Specifically, for any relative price difference $p$ where $profit(p)>0$ and for $t_w<T_w$ and $t<T$, we find that $wait(p,t_w,t)>arb(p,t_w,t)$, i.e., the expected profit from waiting exceeds that from executing the arbitrage immediately.

\paragraph{Historical price change distribution.}

In practice, cryptocurrency prices may deviate from a geometric Brownian motion and exhibit different behaviors.
We incorporate historical price data into our analysis to explore optimal arbitrage extraction under more realistic conditions.

We gathered one month of bid and ask prices for various trading pairs from March 2024.\footnote{We consider perpetual futures prices since Binance -- the largest cryptocurrency exchange -- offers full historical data for their top order book levels while only providing second-by-second price data for spot pairs.
Perpetual futures often have higher trading volumes compared to spot markets, emphasizing that their prices are likely to accurately reflect underlying asset prices.} Using this data, we compute the historical distribution of the change in the mid-price over a time of $\Delta=10$ ms.

We then assume that price changes in each time interval $\Delta$ are independent and follow the historical distribution we obtained. Based on this assumption, we compute the optimal strategy for the time-advantaged arbitrageur.

We perform this analysis for several trading pairs, including Ethereum (ETH-USDT) and Shiba Inu (SHIB-USDT).
In each case, our analysis consistently demonstrates that the optimal strategy is to wait for the full duration of the arbitrageur's time advantage.

\section{Arbitrage Profits}

After analyzing the optimal strategy in a scenario involving a single time-advantaged arbitrageur, we will proceed to estimate the potential arbitrage profits that can be extracted in this context. Furthermore, we will compare these results to other transaction sequencing mechanisms, such as \gls{FCFS} and \gls{PGA}. Specifically, we will examine the following settings:

\paragraph{\glsfirst{FCFS}.}
In this scenario, arbitrageurs compete to be the first to execute a trade when they detect an arbitrage opportunity. This results in all arbitrage opportunities being immediately exploited as soon as they arise. We compute the total arbitrage profit (extracted by all arbitrageurs combined) in this setting by simulating any arbitrage opportunity being captured instantly after it arises.

\paragraph{\glsfirst{PGA}.}
In this setting, all arbitrageurs observe potential opportunities just before the next block is created. They then compete to be the first transaction included in the block to exploit the arbitrage opportunity. In a highly competitive environment, arbitrageurs will bid almost all their potential profits as priority fees, effectively passing almost all arbitrage profits on to the entity that creates the blocks. We simulate this by assuming that all available arbitrage profits are extracted when a new block is generated.
We compute the total arbitrage profits in this scenario by simulating that each time a new block is created, a successful arbitrage extracts the maximal possible profit at this moment. Note that blocks are assumed to be created at regular intervals.

\paragraph{\gls{FCFS} with a time-advantaged actor (Timeboost).}
As explored in the previous sections, it is optimal for the time-advantaged arbitrageur to fully utilize their time advantage and wait precisely this amount before extracting an available arbitrage opportunity.
This scenario is simulated as follows: Whenever a profitable arbitrage opportunity arises, we wait for the exact length of the time advantage and then execute an arbitrage trade that extracts the maximum possible profit.
Note that since the time-advantaged arbitrageur always submits an arbitrage transaction before the delayed transactions of the other arbitrageurs observing the opportunity would be executed, all arbitrage profits are captured by the time-advantaged arbitrageur. Most of these profits would be transferred to the auctioneer in a competitive auction.

\begin{table}[htb]
\centering
  \begin{subtable}[b]{1.0\textwidth}
  \centering
    \begin{tabular}{ |l| c | c | c | } 
     \hline
      & mean & median & 99th percentile \\ \hline
     \hline
     FCFS & 1.69 & 0.27 & 13.62 \\ \hline
     PGA (200ms) & 2.63 & 0.51 & 26.27 \\ \hline
     Timeboost (200ms) & 4.28 & 0.63 & 46.58 \\ \hline
    \end{tabular}
  \caption{ETH-USDT (volatility: 4.2\%, autocorrelation: 2.3\%), pool fee: 0.05\%}
  \label{tab:arb_profits_eth_usdt1}
  \end{subtable}
  \begin{subtable}[b]{1.0\textwidth}
  \centering
    \begin{tabular}{ |l| c | c | c | } 
     \hline
      & mean & median & 99th percentile \\ \hline
     \hline
     FCFS & 0.57 & 0.0 & 6.71 \\ \hline
     PGA (200ms) & 0.89 & 0.0 & 13.16 \\ \hline
     Timeboost (200ms) & 1.82 & 0.0 & 19.47 \\ \hline
    \end{tabular}
  \caption{ETH-USDT (volatility: 4.2\%, autocorrelation: 2.3\%), pool fee: 0.3\%}
  \label{tab:arb_profits_eth_usdt2}
  \end{subtable}
  \begin{subtable}[b]{1.0\textwidth}
  \centering
    \begin{tabular}{ |l| c | c | c | } 
     \hline
      & mean & median & 99th percentile \\ \hline
     \hline
     FCFS & 0.72 & 0.00 & 5.08 \\ \hline
     PGA (200ms) & 0.72 & 0.00 & 5.94 \\ \hline
     Timeboost (200ms) & 0.67 & 0.00 & 6.48 \\ \hline
    \end{tabular}
  \caption{ETH-BTC (volatility: 1.8\%, autocorrelation: -1.1\%), pool fee: 0.05\%}
  \label{tab:arb_profits_eth_btc}
  \end{subtable}
  \begin{subtable}[b]{1.0\textwidth}
  \centering
    \begin{tabular}{ |l| c | c | c | } 
     \hline
      & mean & median & 99th percentile \\ \hline
     \hline
     FCFS & 5.99 & 0.00 & 33.95 \\ \hline
     PGA (200ms) & 5.08 & 0.00 & 42.65 \\ \hline
     Timeboost (200ms) & 5.27 & 0.00 & 45.80 \\ \hline
    \end{tabular}
  \caption{ARB-ETH (volatility: 3.4\%, autocorrelation: -15.3\%), pool fee: 0.05\%}
  \label{tab:arb_profits_arb_eth}
  \end{subtable}
\caption{Minutely arbitrage profits in \$ extracted in each scenario for a constant product \gls{AMM} pool with~\$100,000,000 worth of (active) liquidity.
For each trading pair, we report the daily volatility and the autocorrelation of~50~ms returns.}
\label{tab:arbitrage_profits}
\end{table}

\subsection{Simulation Results}

We perform all simulations using one month of Binance price data from March~2024. Arbitrage profits are calculated for each one-minute interval, which corresponds to the proposed time for which a time advantage is auctioned off under Timeboost's current proposal~\cite{arb_timeboost_2024}.
The simulations are run on a \gls{CPMM} pool with \$100 million in liquidity, or equivalently, a pool utilizing concentrated liquidity with this amount of active liquidity.
This is the order of magnitude of sizes of Uniswap pools for these trading pairs on rollups like Arbitrum at the time of writing.
Specifically, we simulate arbitrage profits for each minute using historical Binance prices for that period. We then assume that the pool initially holds \$100 million in liquidity at the beginning of each minute and is adjusted solely by the arbitrage trades occurring throughout the minute.

Table \ref{tab:arbitrage_profits} reports the average arbitrage profits by minute for several trading pairs.
As expected, profits are higher for more volatile pairs (compare ETH-USDT 0.05\% to ETH-BTC 0.05\%) and lower fee pools (compare ETH-USDT 0.05\% to ETH-BTC 0.3\%).

For the ETH-USDT pair, arbitrageurs earn significantly higher profits with \glspl{PGA} and Timeboost compared to the \gls{FCFS} setting. This aligns with the result from the previous section, where the optimal strategy for a time-advantaged arbitrageur is to wait as long as possible before executing the arbitrage. On the other hand, always arbitraging immediately would lead to the same total arbitrage profits as under \gls{FCFS}.

Note that the analysis deriving the optimal strategy under Timeboost assumes independently distributed price changes. However, when prices exhibit mean reversion -- i.e., negative autocorrelation -- profits under \gls{FCFS} can match or even exceed those from a waiting strategy with Timeboost, as seen with the ETH-BTC and ARB-ETH pairs. In cases of negative autocorrelation, where it is more likely that a price move that opens an arbitrage opportunity is followed by a reversal, it becomes more advantageous to exploit arbitrage opportunities immediately.

Furthermore, Figure \ref{fig:profit_distribution_timeboost} illustrates the distribution of minutely arbitrage profits in the Timeboost scenario. The histogram of profits in Figures \ref{fig:profit_histogram} reveals profits are generally small during most 1-minute intervals, with a long tail of occasional minutes during which extraordinarily large profits occur.
The cumulative profits, shown in Figure \ref{fig:profit_cumulative}, indicate that while profits surge during periods of high price volatility, the majority of gains come from steady growth over time.
The corresponding figures for the \gls{FCFS} and \gls{PGA} settings can be found in Appendix \ref{app:extra_figs} and exhibit very similar characteristics.

\begin{figure}[htp]
\centering
\begin{subfigure}{.5\textwidth}
  \centering
  \includegraphics[width=\linewidth]{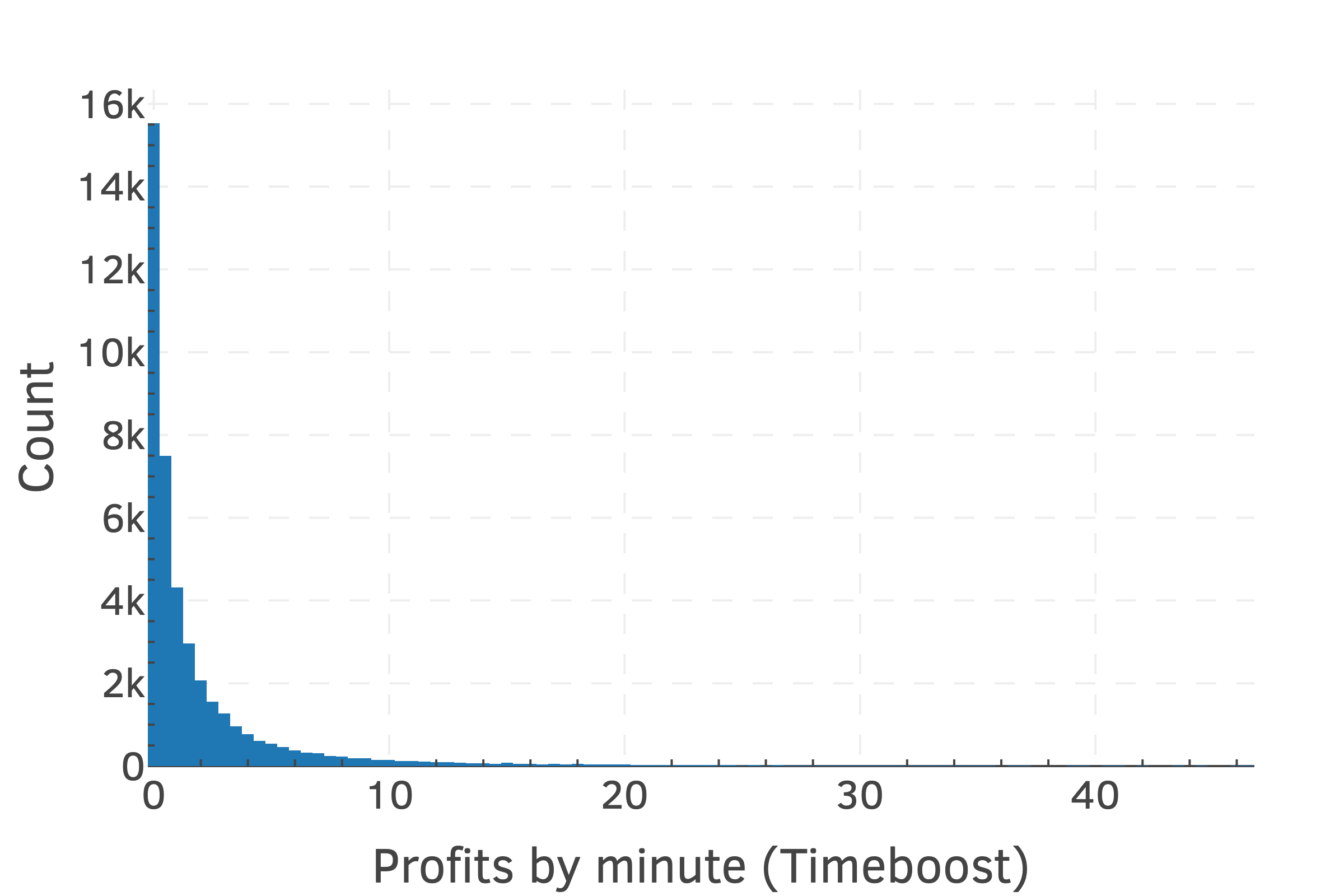}
  \caption{Histogram of minutely profits in \$.}
  \label{fig:profit_histogram}
\end{subfigure}%
\begin{subfigure}{.5\textwidth}
  \centering
  \includegraphics[width=\linewidth]{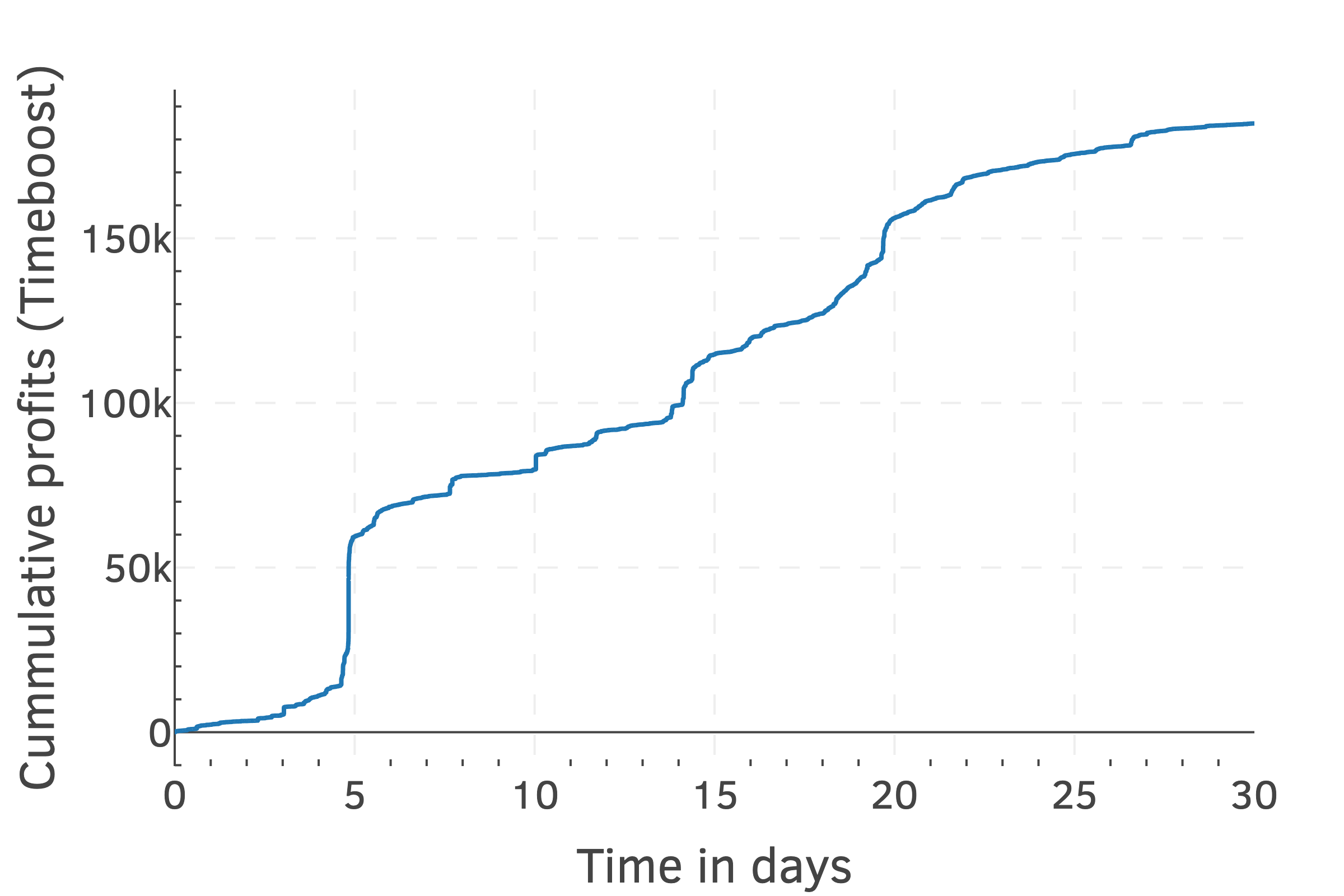}
  \caption{Cumulative profits over time.}
  \label{fig:profit_cumulative}
\end{subfigure}
\caption{Distribution of arbitrage profits for a ETH-USDT 0.05\% pool with \$100m liquidity in the Timeboost setting.}
\label{fig:profit_distribution_timeboost}
\end{figure}

The time advantage in Timeboost and the block time in a \gls{PGA} both represent the maximum time a regular transaction takes to appear on-chain, assuming no congestion.
Table \ref{tab:arb_profits_times} reports the mean arbitrage profits for various values of this parameter. As time intervals increase, profits increase for both \gls{PGA} and Timeboost while maintaining their relative relationship. This increase is consistent with other theoretical \cite{milionis2023automated} and empirical~\cite{fritsch2024arblosses} findings on the relationship between arbitrage profits and block times.

\begin{table}[htp]
    \centering
    \begin{tabular}{ |l| c | c | c | c | c |} 
         \hline
          & 200ms & 500ms & 1s & 2s & 5s \\ \hline
         \hline
         PGA & 2.63 & 3.31 & 3.88 & 4.68 & 5.81 \\ \hline
         Timeboost & 4.28 & 5.65 & 7.01 & 8.43 & 9.54 \\ \hline
         
    \end{tabular}
    \caption{Mean minutely arbitrage profit for a \$100M ETH-USDT pool with a~0.05\% fee for different time parameters}
    \label{tab:arb_profits_times}
\end{table}

\subsection{Capturing Arbitrage Profits for AMMs under the Timeboost Policy}

So far we have explored the size of arbitrage profits and how they are divided between rollups and arbitrageurs under different transaction ordering policies. However, the ultimate goal is arguably to capture and retain these profits at their source -- in the AMM pool.

Attributing \gls{MEV} and returning (part of) it to users is an actively studied area, with multiple theoretical and practical proposals, such as Adaptive Curves~\cite{opt_adapt} and \gls{MEV} Taxes~\cite{robinson_priority_2024}.
Under a time-advantage policy (such as Timeboost), one potential general approach to MEV capture would involve separately auctioning time advantages for interacting with specific contract addresses, such as AMM pools. The proceeds from each auction could then be returned to the relevant contract.

Another approach to mitigating the one-sided value extraction, which also works with the simple Timeboost policy (featuring a single time advantage auction), is to label transactions from the time-advantaged arbitrageur. This enables AMMs to adapt their trading behavior: pool smart contracts could apply different fees to these transactions or adjust the market-making function specifically for them. Additionally, the pool can limit the number of time-advantaged transactions to one per $T_w$.\footnote{Without this limit, arbitrageurs could break up their trades into smaller parts and submit them sequentially. In fact, splitting trades into infinitely many parts would render any adjustments to the market-making function ineffective; see Section 3.2 in \cite{canidio2024arbitrageursprofitslvrsandwich}.}

In the following, we analyze how much arbitrage profit AMMs could capture through this approach. Without any adaptation, the arbitrageur’s profit from moving the pool’s marginal price from $P_0$ to the true price $P$, given pool reserves of $(X,Y)$, equals $profit(P_0,P):=XP-2\sqrt{XYP}+Y$. We assume the pool attempts to capture a $\alpha$ fraction of this profit.
In other words, the AMM will offer a trade that shifts its marginal price to $P$ while incurring a loss of $(1-\alpha)profit(P_0,P)$, for any $P$.
Hence, it allows its reserves to be changed to $(X_1,Y_1)$, such that $\frac{Y_1}{X_1}=P$, and $(X-X_1)P-(Y_1-Y) = (1-\alpha)profit(P_0,P)$.
This fully determines the AMM's trading behavior for time-advantaged transactions. The AMM allows changing its reserves to:
\begin{align*}
    (X_1, Y_1) = \left( (1-\alpha)\sqrt{\frac{XY}{P_1}} + \frac{\alpha}{2}\left(X+\frac{Y}{P_1}\right), (1-\alpha)\sqrt{XYP_1} + \frac{\alpha}{2}\left(XP_1+Y\right) \right)
\end{align*}
However, the time-advantaged arbitrageur, as a best response strategy, may no longer move the current marginal price to the true price $P$ but to some intermediate price $P_1$. By doing so, the time-advantaged arbitrageur makes the following profit:
\begin{multline*}
(X-X_1)P - (Y_1-Y) = \\
Y\left((1-\alpha)p-\frac{\alpha}{2}\frac{p}{p_1} - (1-\alpha)p\frac{1}{\sqrt{p_1}} - (1-\alpha)\sqrt{p_1} - \frac{\alpha}{2}p_1+(1-\frac{\alpha}{2})\right),
\end{multline*}
where $p=\frac{P}{P_0}$ and $p_1=\frac{P_1}{P_0}$. The arbitrageur will choose $P_1$, and thereby $p_1$, to maximize this expression.

Setting $\alpha=0$ corresponds to a standard CPMM pool where the arbitrageur extracts all value, which we call {\it maximum arbitrage profit}. It is equal to $MAP:=Y(p-2\sqrt{p}+1)$. In contrast, when $\alpha=1$, the pool aims to capture all the value. In this case, the time-advantaged arbitrageur's best response is to set $p_1=\sqrt{p}$, yielding a profit of $Y(\frac{1}{2}p-\sqrt{p}+\frac{1}{2}) = \frac{1}{2}MAP$.

For a given $\alpha$ and $p$, the first order condition for finding the best response of the time-advantaged arbitrageur, $p_1$, is: 
\begin{equation*}
    \frac{\alpha p}{2p_1^2}+\frac{(1-\alpha) p}{2 p_1^{3/2}} - \frac{1-\alpha}{2p^{1/2}} -\frac{\alpha}{2}=0.
\end{equation*}
This 4th-degree polynomial with respect to $p_1^{-1/2}$ has an analytical solution. We find that the pool returns are maximized when $\alpha$ is set to $1$, accounting for $\frac{1}{4}MAP$. In this scenario, $\frac{1}{2}MAP$ goes to the time-advantaged arbitrageur, while the remaining $\frac{1}{4}MAP$ stays unrealized, potentially being captured by other arbitrageurs.

\section{Conclusion}

In this work, we study the impact of adding a time advantage to a \gls{FCFS} transaction ordering rule for non-atomic arbitrage and compare the sequencing policies currently employed by the main \glspl{L2}, namely, the original \gls{FCFS} and \glspl{PGA}.

We analyze the optimal strategy for an arbitrageur with a time advantage, which involves deciding whether to arbitrage or wait, depending on which action offers a higher expected profit. When the pool has no trading fees and the expected future price is equal to the current price, we show that the arbitrageur is indifferent between arbitraging and waiting. We extend our analysis by computing the optimal strategy numerically for a scenario with empirical price change distributions and non-zero trading fees. Our findings indicate that the optimal strategy for the arbitrageur is to wait until the time window expires.

Then, we examine arbitrage profits under three sequencing policies -- \gls{FCFS}, \gls{PGA}, and Timeboost -- for several trading pairs. For most pairs, arbitrageurs can extract higher profits under Timeboost than \gls{FCFS} and \gls{PGA}. However, when prices exhibit mean reversion, profits under \gls{FCFS} can match or even exceed those from a waiting strategy with Timeboost. We also observe that more volatile pairs and lower-fee pools offer greater arbitrage opportunities. 

Finally, to mitigate arbitrage profit extraction by arbitrageurs under a time advantage, if the party with the latency advantage can be identified, we propose a mechanism that returns $25\%$ of the available arbitrage gains to the liquidity providers.




%
%
%
\bibliographystyle{splncs04nat}
\bibliography{references}

\appendix

\section{Arbitrage Profits for Constant Product Pools}
\label{annex:cpmm-profit}

Consider a \gls{CPMM} pool with trading fee $f$.
In the following, we calculate the profit of the \emph{maximal arbitrage trade}, that is, the trade that produces the maximum arbitrage profit, depending on the relative price difference between the AMM pool and the external market. Recall that we assume the price impact of the arbitrage trade on the external market to be negligible.

Let $X, Y$ denote the \gls{CPMM}'s reserves. We assume that the token $Y$ is the numéraire asset and measure the arbitrage profits in this token.
For readability, let the tokens in the pool be ETH and USD.
Then, the initial price of ETH on the \gls{CPMM} is $P_0 = \frac{Y}{X}$.

Let $P$ be the external market price, meaning the relative price difference is $p=\frac{P}{P_0}$.
If $(1-f)P_0\leq P\leq \frac{1}{1-f}P_0$, no profitable arbitrage is possible since the marginal price of buying (selling) ETH on the \gls{CPMM} is higher (lower) than the external price.

If $P>\frac{1}{1-f}P_0$, the marginal price of buying ETH on the \gls{CPMM} is lower than the external market price. Hence, there is an arbitrage opportunity.
To exploit this, the arbitrageur will buy $\Delta x$ ETH for $\Delta y$ USD from the \gls{CPMM} and, at the same time, sell $\Delta x$ ETH for $\Delta x\cdot P$ USD on the external market. Note that of the $\Delta y$ USD paid for the \gls{CPMM} trade,  $f\Delta y$ is paid as the trading fee, and  $(1-f)\Delta y$ is used for the swap according to the constant product rule.

For the arbitrage trade yielding the maximal profit, the marginal price of buying ETH from the pool after the trade (considering trading fees) will equal the external price $P$.\footnote{This assumes that the marginal trade price (taking trading fees into account) equals the marginal pool price after the trade (taking trading fees into account), which is true if trading fees are not added to the pool (as in Uniswap v3). Adding trading fees to the pool (as in Uniswap v2) creates a discontinuity in the marginal pool price and leads to more complex expressions.}
Indeed, a lower marginal price would mean that infinitesimally increasing the trade size yields an extra profit, while a higher marginal price implies that the 
last infinitesimal part of the trade comes at a loss, and a smaller arbitrage trade would yield a higher profit.
When the marginal price of buying ETH (including fees) from the pool equals $P$, the pool's marginal price without fees (i.e., the ratio of its reserves) equals $(1-f)P$. 
Hence, the maximal arbitrage trade moves the pool's price to $(1-f)P$, leaving a relative difference of $\frac{1}{1-f}$ to the external price.
Moreover, the pool's reserves after the arbitrage trade are
\begin{align*}
    \left( \sqrt{\frac{XY}{(1-f)P}}, \sqrt{XY (1-f)P} \right).
\end{align*}
This means, the arbitrageur trades
\begin{align*}
    \Delta x &= X - \frac{\sqrt{XY}}{\sqrt{(1-f)P}} \\
    \Delta y &= \frac{1}{1-f}\left(\sqrt{XY}\sqrt{(1-f)P} - Y\right),
\end{align*}
and makes a profit of
\begin{align*}
    \Delta x\cdot P - \Delta y &= XP - \sqrt{XY}\sqrt{\frac{P}{1-f}} - \sqrt{XY}\sqrt{\frac{P}{1-f}} + \frac{1}{1-f}Y \\
    &= Y\left( \frac{P}{P_0} - \frac{2}{\sqrt{1-f}}\sqrt{\frac{P}{P_0}} +\frac{1}{1-f}\right).
\end{align*}
The above formula holds for $P>\frac{1}{1-f}P_0$. For $P<(1-f)P_0$, the marginal price of selling ETH on the AMM is higher than the external market price. Hence, the arbitrageur will buy ETH on the external market and sell it to the AMM.
Arguing analogously to the previous case shows that the maximal arbitrage trade will move the pool's price to $\frac{1}{1-f}P$, leaving a relative difference of $1-f$ to the external price. Hence, the pool's price after the maximal arbitrage is given by
\begin{align*}
    pArb(p) = \begin{cases}
    \frac{1}{1-f}  & \text{if } p > \frac{1}{1-f} \\
    p & \text{if } 1-f < p < \frac{1}{1-f} \\
    1-f & \text{if } p < 1-f
\end{cases}
\end{align*}
Furthermore, similar calculations to the previous case yield that the arbitrageur makes a profit of
\begin{align*}
    Y\left( \frac{1}{1-f}\frac{P}{P_0} - \frac{2}{\sqrt{1-f}}\sqrt{\frac{P}{P_0}} + 1 \right).
\end{align*}
Since $Y$ represents half of the initial pool value in USD, the second term can be interpreted as (two times) the arbitrage profit relative to the pool value.
Therefore, the maximal arbitrage profit relative to the initial pool value in USD, i.e., $2Y$, can be written as a function of the relative price difference $p=\frac{P}{P_0}$:
\begin{align*}
    profit(p) = \begin{cases}
    \frac{1}{2}p - \frac{1}{\sqrt{1-f}}\sqrt{p} +\frac{1}{2}\frac{1}{1-f}  & \text{if } p > \frac{1}{1-f} \\
    0 & \text{if } 1-f < p < \frac{1}{1-f} \\
    \frac{1}{2}\frac{1}{1-f}p - \frac{1}{\sqrt{1-f}}\sqrt{p} +\frac{1}{2} & \text{if } p < 1-f
\end{cases}
\end{align*}

Figure \ref{fig:profit_function} provides an example of the profit function for a trading fee of $f=0.05\%$

Next, we consider the relative change in the pool's value resulting from the maximal arbitrage trade executed at the relative price difference $p$. For $P>\frac{1}{1-f}P_0$, the pool reserves change $(X,Y)$ to $\left( \sqrt{\frac{XY}{(1-f)P}}, \sqrt{XY (1-f)P} \right)$.
In particular, the pool's value at the respective pool price changes from $2Y$ to $2\sqrt{XY(1-f)P}$, i.e.\ by a factor of $\sqrt{1-f}\sqrt{\frac{P}{P_0}}$. Applying the same reasoning to the case where $P < (1-f)P_0$ results in
\begin{align*}
    poolVal(p) = \sqrt{\frac{p}{pArb(p)}} = \begin{cases}
    \sqrt{1-f} \sqrt{p}  & \text{if } p > \frac{1}{1-f} \\
    1 & \text{if } 1-f < p < \frac{1}{1-f} \\
    \frac{1}{\sqrt{1-f}}\sqrt{p} & \text{if } p < 1-f.
\end{cases}
\end{align*}

\section{Extra Figures}
\label{app:extra_figs}

\begin{figure}[ht]
\centering
\begin{subfigure}{.5\textwidth}
  \centering
  \includegraphics[width=\linewidth]{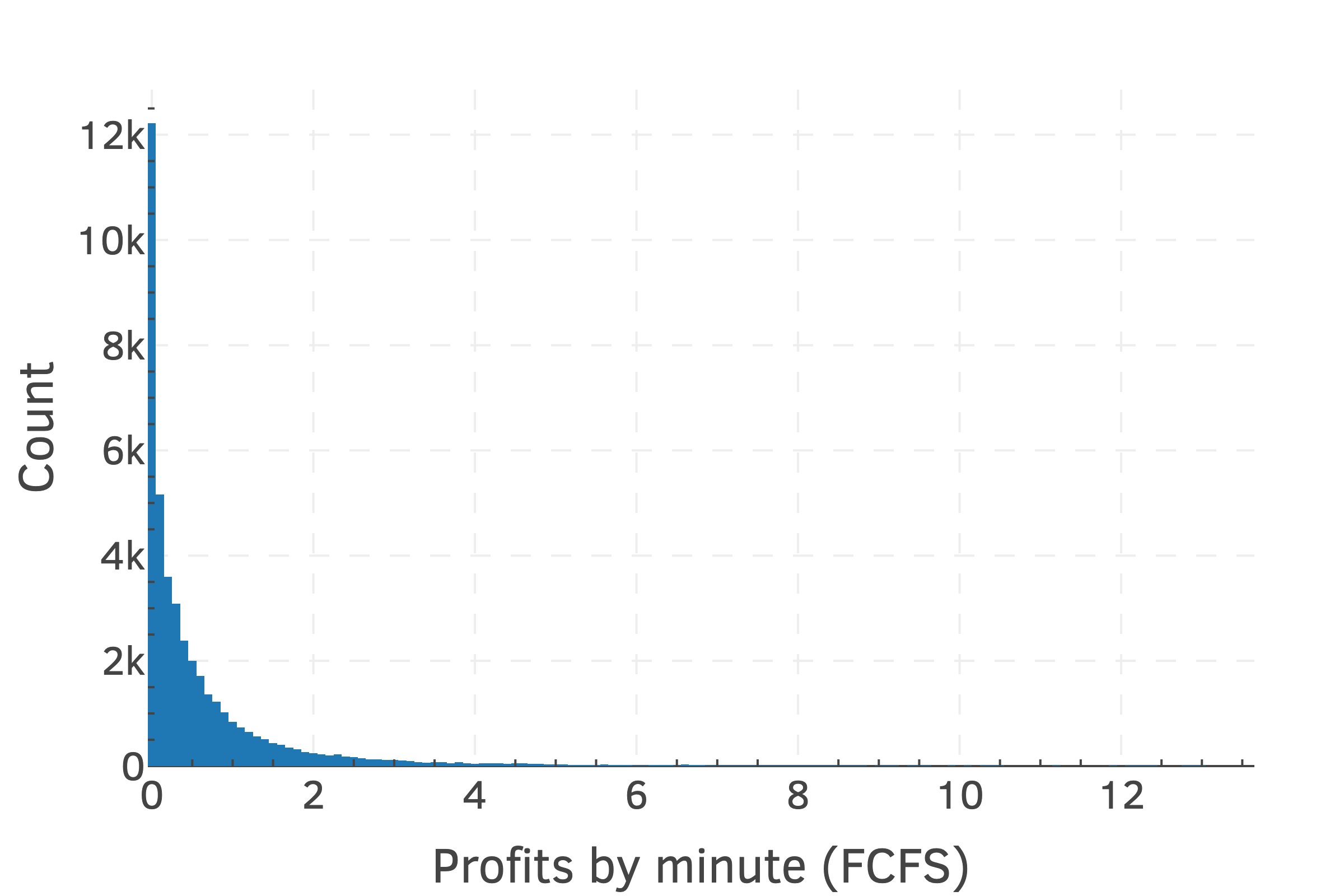}
  \caption{Histogram of minutely profits in \$.}
\end{subfigure}%
\begin{subfigure}{.5\textwidth}
  \centering
  \includegraphics[width=\linewidth]{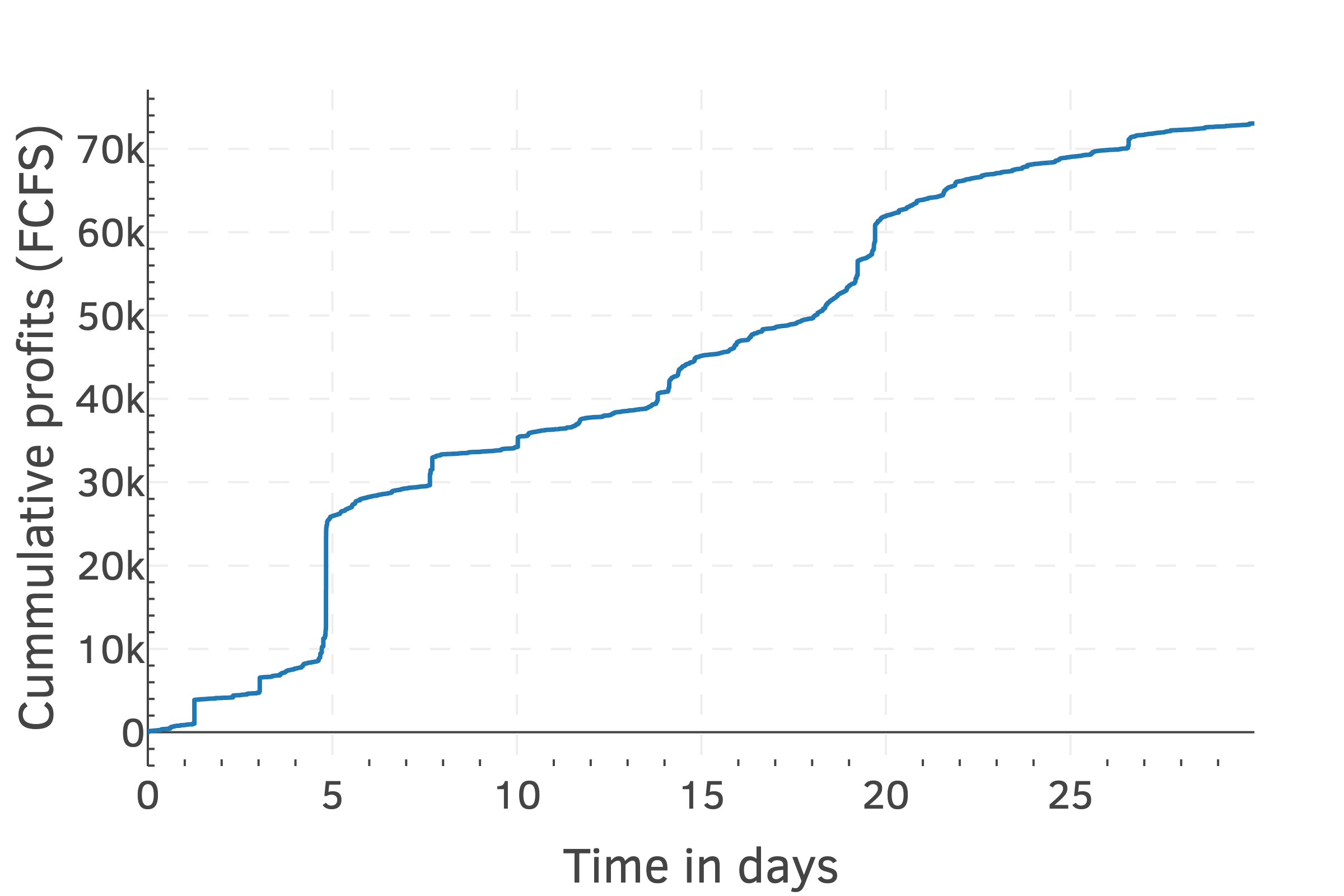}
  \caption{Cumulative profits over time.}
\end{subfigure}
\caption{Distribution of arbitrage profits for a ETH-USDT 0.05\% pool with \$100m liquidity in the first-come-first-serve setting.}
\label{fig:profit_distribution_fcfs}
\end{figure}

\begin{figure}[ht]
\centering
\begin{subfigure}{.5\textwidth}
  \centering
  \includegraphics[width=\linewidth]{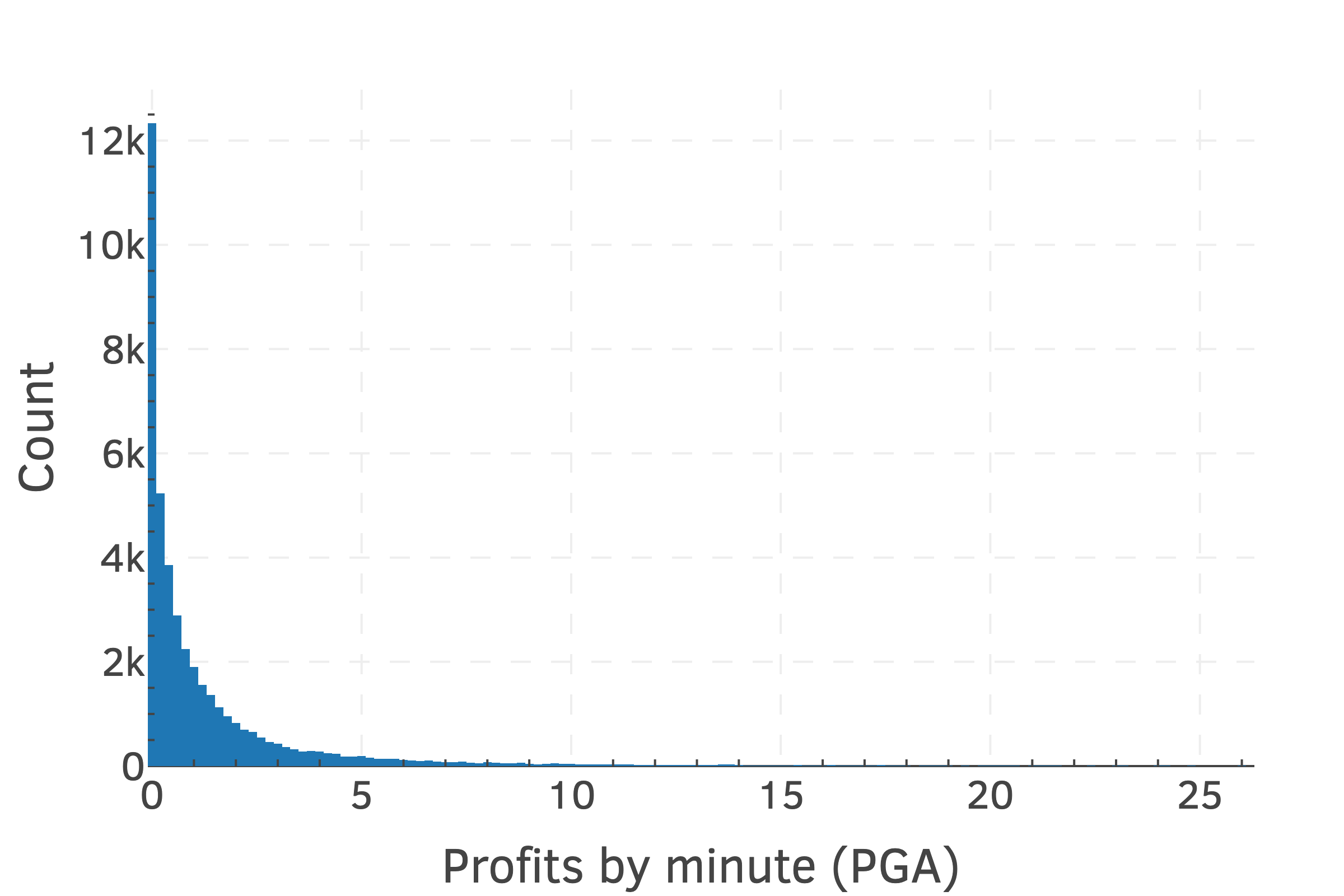}
  \caption{Histogram of minutely profits in \$.}
\end{subfigure}%
\begin{subfigure}{.5\textwidth}
  \centering
  \includegraphics[width=\linewidth]{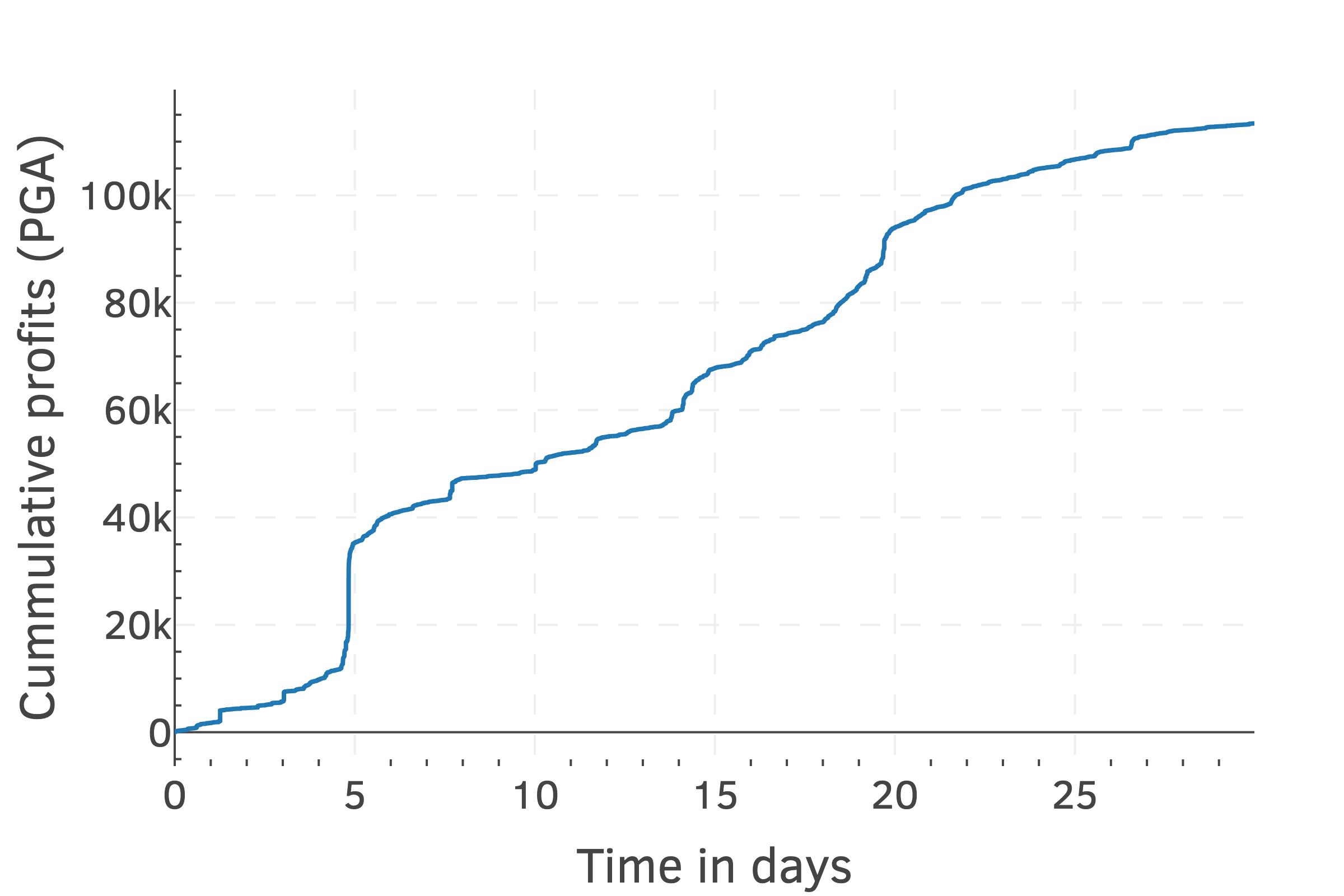}
  \caption{Cumulative profits over time.}
\end{subfigure}
\caption{Distribution of arbitrage profits for a ETH-USDT 0.05\% pool with \$100m liquidity in the priority gas auction setting.}
\label{fig:profit_distribution_priority}
\end{figure}

\end{document}